\documentclass[11pt,a4paper]{article}
\usepackage{amssymb}
\usepackage{amsmath}
\usepackage{latexsym}
\usepackage{graphicx}
\usepackage{color}
\usepackage{hyperref}
\newcommand{\boldgreek}[1]{\mbox{\boldmath$#1$}}
\newcommand{\nfrac}[2]{{\textstyle \frac{#1}{#2} }}
\newcommand{\R}{I\kern-0.37emR}

\newcommand{\ny}{n\rightarrow\infty}
\newcommand{\Q}{I\kern-0.37emP}
\newcommand{\E}{I\kern-0.37emE}
\newtheorem{THE}{Theorem}[section]

\begin{document}
\title{R-estimation in a Linear Model with Autoregressive Errors}

\author{
Jana Jure\v{c}kov\'a\thanks{The research of Jana Jure\v{c}kov\'a and Jan Picek was supported by Grant 22-03636S of the Czech Science Foundation.} $^{\small 1}$, 
Hira L. Koul $^{\small 2}$, 
Jan Picek\thanks{Corresponding author. Email: jan.picek@tul.cz} $^{\small 3}$\\[2ex]
$^{\small 1}$ Institute of Information Theory and Automation, Czech Academy of Sciences,\\
Pod vod\'arenskou v\v{e}\v{z}\'i 4, 182 00 Prague, Czech Republic\\
Email: jureckova@utia.cas.cz\\[1ex]
$^{\small 2}$ Department of Statistics and Probability, Michigan State University,\\
East Lansing, MI 48824, USA\\
Email: koul@msu.edu\\[1ex]
$^{\small 3}$ Technical University of Liberec, Faculty of Science, Humanities and Education,\\
Studentsk\'a 2, 461 17 Liberec, Czech Republic\\
Email: jan.picek@tul.cz
}

\date{} % Empty date

\maketitle

\begin{abstract} 
In the linear regression model, we construct a nonparametric estimate of the regression parameter vector $\boldgreek\beta$ that is insensitive to a possible nuisance autoregression in the model errors. The main tool for estimating $\boldgreek\beta$ is based on the autoregression rank scores of the model. The resulting estimator is invariant to the autoregression parameters and thus remains insensitive to potential hidden linear trends or other structured disturbances, which frequently occur in economic, hydrological, and related applications.
\end{abstract}

\section{Introduction}
\setcounter{equation}{0}
Classical regression methods, such as Ordinary Least Squares (OLS), often fail in the presence of outliers or autocorrelated error structures. These conditions are common in practice, yet often remain undetected in observational data. As a result, standard estimation procedures can lead to misleading conclusions — for instance, they may mask an upward trend in hydrological or economic applications.

To address this, robust alternatives are needed that remain valid under such departures from ideal assumptions. In particular, the presence of autocorrelation in model errors can distort inference on regression parameters unless explicitly accounted for. In this paper, we propose a novel rank-based estimator that is invariant to both nuisance autoregression parameters and the intercept, while preserving desirable asymptotic properties. This estimator builds upon and extends the ideas of Jaeckel~\cite{Jaeckel1972} and Koul and Saleh~\cite{KoulSaleh1995}, and is particularly suited to models where errors follow an unknown autoregressive structure.

We consider a linear regression model of order $s$, where the model errors follow a stationary autoregressive process of order~$p$:
\begin{eqnarray}
\label{1}
&&y_{t} = \beta_0 + \mathbf{x}_t^{\top}\boldgreek\beta + \varepsilon_t 
= \beta_0 + x_{t1}\beta_1 + \ldots + x_{ts}\beta_s + \varepsilon_t, \\[2mm]
&&\varepsilon_t = \phi_0 + \phi_1 \varepsilon_{t-1} + \ldots + \phi_p \varepsilon_{t-p} + u_t, \quad t = 1, 2, \ldots, n, \nonumber \\[2mm]
&&\boldgreek\beta = (\beta_1, \ldots, \beta_s)^{\top}, \quad \boldgreek\beta^{\ast} = (\beta_0, \beta_1, \ldots, \beta_s)^{\top}. \nonumber
\end{eqnarray}

Here, $y_t$ is the response variable, $\mathbf{x}_t = (x_{t1}, \ldots, x_{ts})^{\top}$ are the regressors, and $\beta_j$, $j = 1, \ldots, s$ are the unknown regression parameters, with $\beta_0$ as the intercept. Moreover, $\phi_0, \phi_1, \ldots, \phi_p$ are the unknown autoregressive parameters, where $\phi_0$ is included for mathematical convenience and may be zero. The parameters $\beta_0$ and $\boldgreek\beta$ are of primary interest, while $\phi_0, \ldots, \phi_p$ are considered nuisance parameters. The innovations $u_t$ are assumed to be independent and identically distributed (\textit{i.i.d.}) with a continuous distribution function $F(u)$, increasing on the set $\{u : 0 < F(u) < 1\}$, and with density $f(u)$, which is generally unknown but satisfies:
\begin{equation}\label{moments_of_u}
\E(u_t) = 0, \quad \mathrm{Var}(u_t) = \sigma^2 < \infty.
\end{equation}

We assume that the error process given in equation~(\ref{1}) is strictly stationary, so that the errors $\varepsilon_t$ share a common marginal distribution and, in particular, the same quantiles.  The observable variables are $y_t$, for $t = 0, \pm1, \pm2, \ldots$. For identifiability, we assume that the initial values $(y_{-p+1}, \ldots, y_0)$ are known.

Hallin and Jure\v{c}kov\'a~\cite{Hallin1999} considered the autoregressive model without a nuisance linear regression and constructed an asymptotically optimal rank test for hypotheses concerning the autoregression parameters $\phi_0, \ldots, \phi_p$. If the autoregression is treated as a nuisance and our inference focuses on the linear regression, then its presence can distort the conclusions drawn from the data. For example, it may lead to a false detection of a linear trend in water flow, or may mask a true trend.

Keeping this in mind, Jure\v{c}kov\'a et al.~\cite{Metrika} proposed a nonparametric test of the hypothesis $\mathbf{H}_0: \boldgreek\beta = \mathbf{0}$ that is invariant to the nuisance autoregression. Their main tool was the concept of \textit{autoregression rank scores}, introduced by Koul and Saleh~\cite{KoulSaleh1995} as the dual counterpart to the autoregression quantile. These rank scores themselves are invariant to the nuisance autoregression, which ensures the invariance of the test statistic's distribution under the null hypothesis.

However, the invariance of the test criterion under the null hypothesis is not sufficient for constructing a point estimator of $\boldgreek\beta$ when other parameters remain unspecified. Therefore, it is necessary to study the behavior of the test criterion in a neighborhood of the null hypothesis. This is addressed in the present paper, where we show that the linear autoregression rank statistic is asymptotically linear in the regression parameter. This property enables us to construct an invariant estimator of $\boldgreek\beta$, which corresponds to a nonparametric rank estimator (R-estimator). To adapt the rank-based approach to the autoregressive setting, the ranks of the residuals $y_t - \mathbf{x}_t^{\top} \mathbf{b}$ are replaced by the autoregression rank scores of these residuals. The resulting estimators of $\boldgreek\beta$ will be referred to as \textit{RR-estimators}.

The problem of estimating the autoregression parameters themselves is not addressed in this paper. It has been studied, for example, in~\cite{Koul Ossiander1994}, \cite{KoulSaleh1995}, \cite{Koul2002}, \cite{Sankhya2005}, among others — typically without considering an additional linear regression. A linear programming estimator of the autoregression parameters, convenient in cases with heavy-tailed errors, was proposed in~\cite{Feigin}. With respect to hypothesis testing, \cite{STPA} proposed a test for serial independence in a linear model against potential autoregressive error structures. However, the problem of point estimation of autoregression parameters in the presence of an additional nuisance regression remains open, despite its clear relevance in applications.

This paper is organized as follows. Section~2 briefly reviews the concepts of autoregression quantiles and autoregression rank scores, along with their key properties (for further details, see~\cite{KoulSaleh1995}, \cite{Hallin1999}, \cite{Guney1}, \cite{Aplimat2020}), and , \cite{Jur_Slopes}. In Section~3, we define the RR-estimator of the regression parameter $\boldgreek\beta$ and illustrate its properties. Analogously to the rank estimators in models with \textit{i.i.d.} errors, the RR-estimator is defined as a minimizer of a specific form of Jaeckel’s rank dispersion measure~\cite{Jaeckel1972}.

Section~4 presents a simulation study that compares the proposed RR-estimator with OLS and median quantile regression under various error distributions, including normal, contaminated normal, and Cauchy. Section~5 provides a numerical illustration using real-world data on real estate prices, demonstrating the robustness and practical applicability of the RR-estimator in comparison to classical and quantile-based methods.

\section{Autoregression rank scores and autoregression quantiles}
\setcounter{equation}{0}

Assume first that the observations $y_t$ follow the autoregressive (AR) model
\begin{equation}
	y_t = \phi_0 + \phi_1 y_{t-1} + \ldots + \phi_p y_{t-p} + u_t, \quad t \in \{0,\, \pm1,\, \pm2,\, \ldots\},
	\label{2.1}
\end{equation}
with \textit{i.i.d.} errors $u_t$. Denote the vectors
\[
\mathbf{y}_t = (y_t, \ldots, y_{t-p+1})^{\top}, \quad \mathbf{y}_t^* = (1,\, \mathbf{y}_t^{\top})^{\top}, \quad t = 0, \ldots, n-1,
\]
and consider the ($n\times p$ and $n\times (p+1)$, respectively) random matrices
\[
\mathbf{Y}_n = \left( \mathbf{y}_0, \ldots, \mathbf{y}_{n-1} \right)^{\top}, \quad \mathbf{Y}_n^* = \left( \mathbf{y}_0^*, \ldots, \mathbf{y}_{n-1}^* \right)^{\top}.
\]

The $\alpha$-\textit{autoregression quantile} is defined as the solution
\[
\widehat{\boldgreek\rho}_n(\alpha) = \left( \widehat{\rho}^0_n(\alpha),\, \widehat{\boldgreek\rho}^{1\top}_n(\alpha) \right)^{\top},
\]
where $\widehat{\rho}^0_n(\alpha) \in \mathbb{R}$ and $\widehat{\boldgreek\rho}^{1}_n(\alpha) \in \mathbb{R}^p$, obtained by minimizing over $\mathbf{r} = (r_0, \mathbf{r}_1^{\top})^{\top} \in \mathbb{R}^{p+1}$ the objective function
\begin{equation}
\sum_{t=1}^{n} h_{\alpha}(y_t - r_0 - \mathbf{y}_{t-1}^{\top}\mathbf{r}_1),
\quad \text{where} \quad
h_{\alpha}(u) = |u|\left( \alpha I[u>0] + (1-\alpha)I[u<0] \right).
\label{2.5}
\end{equation}

Let $\boldgreek\Sigma_n$ and $\boldgreek\Sigma_n^*$ denote the empirical covariance matrices
\begin{equation}
\label{35}
\boldgreek\Sigma_n = \frac{1}{n} \sum_{t=1}^n \mathbf{y}_{t-1}\mathbf{y}_{t-1}^{\top}, \qquad
\boldgreek\Sigma_n^* = \frac{1}{n} \sum_{t=1}^n \mathbf{y}_{t-1}^* (\mathbf{y}_{t-1}^*)^{\top}.
\end{equation}
Under regularity conditions \textbf{(F1–F2)} and \textbf{(X1–X3)} stated below, there exist positive definite matrices $\boldgreek\Sigma$ and $\boldgreek\Sigma^*$ of respective  orders $p\times p$ and $(p+1)\times(p+1),$ such that 
\begin{equation}
\label{35a}
\boldgreek\Sigma_n \stackrel{p}{\longrightarrow} \boldgreek\Sigma,
\qquad
\boldgreek\Sigma_n^* \stackrel{p}{\longrightarrow} \boldgreek\Sigma^*,
\quad \text{as} \quad n \to \infty.
\end{equation}

The definition in (\ref{2.5}) formally coincides with the classical regression quantile of Koenker and Bassett. Analogously, the dual linear programming formulation can be written as:
\begin{equation}
\label{H2.8}
\begin{array}{ll}
\text{maximize} & \mathbf{y}_n^{\top}\mathbf{a} \\
\text{subject to} & \mathbf{1}_n^{\top}\mathbf{a} = n(1-\alpha), \\
& \mathbf{Y}_n^{*\top}(\mathbf{a} - (1-\alpha)\mathbf{1}_n) = \mathbf{0}, \\
& \mathbf{a} \in [0,1]^n, \quad 0 \leq \alpha \leq 1.
\end{array}
\end{equation}

The solution $\widehat{\mathbf{a}}_n(\alpha)$, referred to as the vector of \textit{autoregression rank scores (ARS)}, is autoregression
invariant, hence independent of the  nuisance autoregression parameters$\phi_1, \ldots, \phi_p$. In the special case where $\phi_1 = \ldots = \phi_p = 0$, the components $\widehat{a}_t(\alpha)$ coincide with the classical H\'ajek rank scores~\cite{Hajek1965}:
\[
a^*_t(\alpha) =
\begin{cases}
1, & \text{if} \quad 0 \leq \alpha < \frac{R_{n,t} - 1}{n}, \\
R_{n,t} - n\alpha, & \text{if} \quad \frac{R_{n,t} - 1}{n} \leq \alpha < \frac{R_{n,t}}{n}, \\
0, & \text{if} \quad \frac{R_{n,t}}{n} \leq \alpha \leq 1,
\end{cases}
\]
where $R_{n,t}$ denotes the rank of $y_t$ among $\{y_1, \ldots, y_n\}$.

Since $\widehat{\boldgreek\rho}_{n}(\alpha)$ and $\widehat{\mathbf{a}}_n(\alpha)$ are dual to each other, their relationship satisfies:
\begin{equation}
\label{RR1}
\widehat{a}_{n,t}(\alpha) =
\begin{cases}
1, & \text{if} \quad y_t > \mathbf{y}_t^{\top}\widehat{\boldgreek\rho}_n(\alpha), \\
\text{between } 0 \text{ and } 1, & \text{if} \quad y_t = \mathbf{y}_t^{\top}\widehat{\boldgreek\rho}_n(\alpha), \\
0, & \text{if} \quad y_t < \mathbf{y}_t^{\top}\widehat{\boldgreek\rho}_n(\alpha),
\end{cases}
\quad t=1,\ldots,n.
\end{equation}

To study the behavior of these quantities, we assume that the unknown density $f$ of $u_t$ belongs to the family $\mathcal{F}$ of exponentially tailed distributions satisfying:

\begin{enumerate}
\item[\textbf{(F1)}] $f$ is positive, absolutely continuous, and has finite Fisher information:
\[
\mathcal{I}(f) = \int \left( \frac{f'(x)}{f(x)} \right)^2 f(x) \, dx < \infty.
\]
Additionally, there exists $K_f \geq 0$ such that $f'$ and $f''$ are bounded for all $|x| > K_f$.

\item[\textbf{(F2)}] $f(x) \to 0$ monotonically as $x \to \pm\infty$ and the distribution function $F$ satisfies:
\[
\lim_{x \to -\infty} \frac{-\log F(x)}{c|x|^r} =
\lim_{x \to \infty} \frac{-\log(1-F(x))}{c|x|^r} = 1,
\]
for some $c > 0$ and $r \geq 1$.
\end{enumerate}

Further consequences of (F1) and (F2) are detailed in~\cite{Hallin1999} and~\cite{Guney1}.

\smallskip

In what follows, we denote by $\mathbf{X}_n$ the $n \times s$ regression  matrix formed by stacking the regressors from model~(\ref{1}), that is,
\[
\mathbf{X}_n = \begin{pmatrix}
\mathbf{x}_1^{\top} \\
\mathbf{x}_2^{\top} \\
\vdots \\
\mathbf{x}_n^{\top}
\end{pmatrix}, \quad \text{with } \mathbf{x}_t = (x_{t1}, \ldots, x_{ts})^{\top}.
\]

Moreover, the following standard conditions are imposed on this matrix $\mathbf{X}_n$:
\begin{enumerate}
\item[\textbf{(X1)}] The matrix $\mathbf{A}_n = n^{-1} \mathbf{X}_n^{\top} \mathbf{X}_n$, of order $s \times s$, is positive definite for $n \geq n_0$.

\item[\textbf{(X2)}] $n^{-1} \sum_{t=1}^n \|\mathbf{x}_{t}\|^4 = O(1)$ as $n \to \infty$.

\item[\textbf{(X3)}] $\lim_{n \to \infty} \max_{1 \leq t \leq n} \left\{ n^{-1} \mathbf{x}_{t}^{\top} \mathbf{A}_n^{-1} \mathbf{x}_{t} \right\} = 0$.
\end{enumerate}

Let $\mathbf{D}_n$ and $\mathbf{Q}_n$ denote the following two matrices, of respective orders $(s+1)\times (s+1)$ and $s\times s$:
\begin{eqnarray}\label{2.3a}
\mathbf{D}_n &=& n^{-1} \mathbf{Y}_n^{\top} \mathbf{Y}_n, \\
\mathbf{Q}_n &=& n^{-1} (\mathbf{X}_n - \overline{\mathbf{X}}_n)^{\top} (\mathbf{X}_n - \overline{\mathbf{X}}_n). \nonumber
\end{eqnarray}
We assume that
\begin{equation}
\label{2.3b}
\lim_{n \to \infty} \mathbb{E}(\mathbf{D}_n) = \mathbf{D}, \qquad
\lim_{n \to \infty} \mathbb{E}(\mathbf{Q}_n) = \mathbf{Q},
\end{equation}
where $\mathbf{D}$ and $\mathbf{Q}$ are positive definite matrices of the same respective orders.

\section{RR-estimator and classical R-estimator of $\boldgreek\beta$}
\setcounter{equation}{0}

Our  goal is to estimate the parameter $\boldgreek\beta$ by  a suitable rank estimate. To proceed, 
fix an $\alpha\in(0,1)$, possibly around 1/2, and consider a nondecreasing score generating function 
\begin{equation}
\label{8c}
\varphi_{\alpha}(u)=\alpha-\mathbb I[u<\alpha], \; \ 0<u<1.
\end{equation}
If we are not aware of  autocovariances between the model errors, we  normally use the nonparametric
R-estimator of $\boldgreek\beta$, introduced in  \cite{Jur1971} and \cite{Koul1971};  for its further study we refer to \cite{JurSen1996}, 
\cite{Sankhya2005}, \cite{JurSenPicek}, among others. 
Let $R_{nt}({\mathbf y}-{\mathbf X}{\mathbf b})$  denote  the rank of  residual
$y_t-{\mathbf x}_t^{\top}{\mathbf b}$, among
$\left(y_1-{\mathbf x}_1^{\top}{\mathbf b},\ldots,y_n-{\mathbf x}_n^{\top}{\mathbf b}\right)$, for ${\mathbf b}\in{\R}^{s}$. 
The R-estimator, denoted   $\check{\boldgreek\beta}_{nR}(\alpha)$, is defined %through the Jaeckel criterion introduced in \cite{Jaeckel1972}; 
 as a  minimizer   of  the following Jaeckel \cite{Jaeckel1972} 
criterion ${\cal D}_n$,  over ${\mathbf b}\in{\R}^{s}$:  %with a fixed  $\alpha\in (0,1)$ and  with the score-generating function $\varphi_{\alpha}$: 
\begin{equation}
\label{11a}
{\mathcal D}_n({\mathbf b})=\sum_{t=1}^n(y_t-{\mathbf x}_t^{\top}{\mathbf b})%
\left[\varphi_{\alpha}\left(\frac{R_{nt}({\mathbf y}-{\mathbf X}{\mathbf b})}{n+1}\right)-\overline{\varphi}_{n\alpha}\right].
\end{equation}
Such estimator $\check{\boldgreek\beta}_{nR}(\alpha)$ in fact estimates only the slope components\\
$\boldgreek\beta=(\beta_1,\ldots,\beta_s)^{\top}.$
It is invariant to $\beta_0$ and does not depend on the value of  quantile $F^{-1}(\alpha).$
  The additional  intercept component of the estimator will be  the $[n\alpha]$-quantile $\check{\beta}_{n0}$  of 
$y_1-\mathbf x_1^{\top}\check{\boldgreek\beta}_{nR,\alpha},\ldots, y_n-\mathbf x_n^{\top}\check{\boldgreek\beta}_{nR,\alpha} $. 

However, such estimator is still  affected by the possible nuisance autoregression of model errors. In order to reduce  this effect,  
we get an inspiration from the test of the  hypothesis $\mathbf H_0: \{\boldgreek\beta=\mathbf 0\}$, developed in \cite{Metrika}; its test criterion is based 
on the linear autoregression rank score statistic. 
   We can consider the test of hypothesis $\mathbf H_{\mathbf b}: \{\boldgreek\beta=\mathbf b\}$ analogously,  for any fixed $\mathbf  b \in \R^s$.  
Such test would be  based on the ranks of residuals $\tilde{y}_{t,\mathbf b}=y_t-\mathbf x_t^{\top}\mathbf b,$  which under  validity of the hypothesis
$\mathbf H_{\mathbf b}$ 
would satisfy $\tilde{y}_{t,\mathbf b} =\beta_0+\varepsilon_t,  \;  t=0,\pm 1, \pm 2,\ldots.$  
Because   the  $\varepsilon_t, \; t=0,\pm1 \pm 2,\ldots$ follow the autoregression  model of the assumed order $p$, 
 the residuals $\tilde{y}_{t,\mathbf b} $ will follow, under validity of $\mathbf H_{\mathbf b}$,  the model 
\begin{equation}\label{11b}
\tilde{y}_{t,\mathbf b} = \beta_0 + \phi_0 + \phi_1\tilde{y}_{t-1,\mathbf b} +\ldots + \phi_p \tilde{y}_{t-p,\mathbf b} + u_t, \quad t = 1, ··· , n.
\end{equation}
 To profit from the invariance of the autoregression rank scores, we use them instead of the ranks of residuals. 
Hence we calculate  the $\alpha$-autoregression rank scores 
$$\widehat{\mathbf a}_{n\alpha}(\mathbf b)=(\hat{a}_{1\alpha}(\mathbf b),\ldots,\hat{a}_{n\alpha}(\mathbf b))^{\top}$$
for the residuals/pseudo-observations $\tilde{y}_{t,\mathbf b}=y_t-\mathbf x_t^{\top}\mathbf b$, ongoingly for $\mathbf b\in \R^s,   \quad t=1,\ldots,n.$
 As an estimator of $\boldgreek\beta$ we propose
  $\widetilde{\boldgreek\beta}_{n\alpha}$ of  the following modification of the Jaeckel measure of the rank dispersions:
\begin{equation}\label{11c}
\mathcal D_{n\alpha}(\mathbf b) =\sum_{t=1}^n(y_t-\mathbf x_t^{\top}\mathbf b)\left(\hat{a}_{t\alpha}(\mathbf b)-\bar{a}_{n\alpha}\right) =\min, \;  \; 
\mathbf b\in \R^s.
\end{equation}
Every point $\mathbf b=\widetilde{\boldgreek\beta}_{nR}(\alpha)$ satisfying (\ref{11c})  is considered as 
a possible  estimator of $\boldgreek\beta.$ 
The measure $\mathcal D_n(\mathbf b)$, based on the autoregression rank scores, is invariant to the autoregression parameters. 
Following \cite{Jaeckel1972}, we see that it is a continuous, convex and piecewise linear function of $\mathbf b\in \R^s$. 
As such, it is differentiable in $\mathbf b$ a.e.;
we can write
\begin{equation} \label{3.6}
\frac{\partial}{\partial \mathbf b}\mathcal D_n(\boldgreek\beta+n^{-1/2}\mathbf b)\Big |_{\mathbf b_0}
= n^{-1/2}(\mathbf X_n-\overline{\mathbf X}_n)^{\top}\widehat{\mathbf a}_{n\alpha}(\mathbf y_n+n^{-1/2}\mathbf b_0))
\end{equation}\
at any point of differentiability of $\mathcal D_n$. 

Next we proceed as follows: Consider the residuals $\tilde{\mathbf y}_{n,\mathbf b}=\mathbf y_n-\mathbf X_n\mathbf b, \; \mathbf b\in \R^s$, 
and for them calculate  the $\alpha$-regression quantile and  $\alpha$-regression rank scores.
These two entities are dual to each other, hence they are in the  relation
$$
\hat{a}_{n\alpha,t}(\mathbf y_n-\mathbf X_n\mathbf b )= \left\{     
\begin{array}{lll}
\; 1          &  \ldots  &\tilde{y}_{t,\mathbf b} > \mathbf Y_{t-1}\widehat{\boldgreek\rho}_{n\alpha}(\tilde{\mathbf y}_{n}(\mathbf b))\\[2mm]
\in(0,1)   &  \ldots  &\tilde{y}_{t,\mathbf b} = \mathbf Y_{t-1}\widehat{\boldgreek\rho}_{n\alpha}(\tilde{\mathbf y}_{n}(\mathbf b))\\[2mm]
\; 0         &   \ldots  &\tilde{y}_{t,\mathbf b} < \mathbf Y_{t-1}\widehat{\boldgreek\rho}_{n\alpha}(\tilde{\mathbf y}_{n}(\mathbf b)), \quad t = 1, ..., n.\\ 

\end{array}
\right . 
$$

 We propose the estimator $\widetilde{\boldgreek\beta}_{nRR}(\alpha)\in{\R}^{s}$  of $\boldgreek\beta$, as a minimizer of the modification of Jaeckel
\cite{Jaeckel1972} measure of the rank dispersion, with the $\alpha$-autoregression rank scores inserted;  the modified measure has the form
\begin{eqnarray}\label{Jaeckel} 
&\mathcal{D}_n(\mathbf b)=\min, \; \mathbf b\in \R^s, &\\ [2mm]
\mbox{ where } & \mathcal{D}_n(\mathbf b)=\sum_{t=1}^n(y_t -\mathbf x_t^{\top}\mathbf b)[-\hat{a}_{t,\alpha}(\mathbf y_n-\mathbf X_n\mathbf b)
+\bar{a}_{n\alpha}].&\nonumber  
\end{eqnarray}

Note that  
 $\bar{a}_{n\alpha}=\frac 1n \sum_{t=1}^n \hat{a}_{t\alpha}(\mathbf y_n-\mathbf x_n^{\top}\mathbf b)$ is constant in $\mathbf b$,  
as an average of the AR rank scores.
It implies that the solution of (\ref{Jaeckel}) is invariant to the intercept and to the nuisance autoregression parameters $\phi_1,\ldots,\phi_p$. 

The subgradient of (\ref{Jaeckel}) is the \textit{linear autoregression rank score statistic}
\begin{eqnarray}
\label{37}
\mathbf S_n(\mathbf y_n-\mathbf X_n{\mathbf b})&=&(S_{n,1}(\mathbf y_n-\mathbf X_n\mathbf b),\ldots,
S_{n,s}(\mathbf y_n-\mathbf X_n\mathbf b)^{\top}\nonumber\\[3mm] 
&=&-n^{-1}\left[\mathbf X_n-\bar{\mathbf X}_n\right]\hat{\mathbf a}_{n\alpha}(\mathbf y_n-\mathbf X_n\mathbf b)\\  %\nonumber\\
S_{n,j}(\mathbf y_n-\mathbf X_n\mathbf b)&=&-n^{-1}\sum_{t=1}^n (x_{t,j}-\bar{x}_{n,j})\hat{a}_{n\alpha,t}(\mathbf y_n-\mathbf X_n\mathbf b),  
 \;    j=1,\ldots,s. \nonumber
\end{eqnarray}
Hence, the estimate $\widetilde{\boldgreek\beta}_{n,RR}$ of $\boldgreek\beta$ can be alternatively expressed as a solution of the minimization
\begin{equation}
\label{38}
\|{\mathbf S}_n(\mathbf y-\mathbf X_n\mathbf b)\|:=\min, \quad {\mathbf b}\in{\R}^s.
\end{equation}
Because the building stones of  $\widetilde{\boldgreek\beta}_{n,RR}$ are invariant to the autoregression parameters and to the intercept, 
the resulting estimator is equally invariant. 
We can verify its consistency and other asymptotic properties, following up the approach  for R-estimation in \cite{JurNonpar92} 
and \cite{JurSenPicek}, with application of the results in \cite{KoulSaleh1995} and \cite{Hallin1999}. 
This is detailed below. The method is based on the asymptotic 
linearity of the autoregression rank score statistics.

\subsection{Asymptotic linearity of autoregression rank scores}

The RR-estimator $\widehat{\boldgreek\beta}_{n\alpha}$ of the regression parameter is a minimizer of the modified Jaeckel measure (\ref{11c})
or a minimizer of the norm of the linear autoregression rank score statistic (\ref{37}).  
The  invariance of the autoregression rank scores guarantees the invariance of the RR-estimator with respect to 
the nuisance autoregression parameters. The consistency of RR-estimator is proven with the aid of  
the uniform asymptotic linearity  of the autoregression rank scores or of the  RRS 
statistic ${\mathbf S}_n({\mathbf b})$, in the regression parameter. This asymptotic linearity in turn implies the asymptotic quadraticity of the 
criterion (\ref{11c}) in the regression parameter.  
The following proposition  on the asymptotic linearity of the linear autoregression rank score statistics
 ${\mathbf S}_n({\mathbf b})$,  is due to Koul and Saleh  \cite{KoulSaleh1995}. Their conditions on $F$ and on the covariates $\mathbf X_n$
can be apparently modified, depending on the  application in the real situation  (e.g. \cite{Hallin1999}). 
%proved  utilizing the ranks of the residuals   under some conditons  
 %of $A_n(\alpha, \mathbf z)$ on $\mathbf d$  
 %has been proven by or of the corresponding RRS  ${\mathbf S}_n({\mathbf b})$, see also
%the analogous results  for linear regression with \textit{i.i.d.} errors in \cite{GJKP},  \cite{Jur/Saleh92}, \cite{JurNonpar92}. 
%with the use of results in \cite{Koul Ossiander1994}. 
 %The weak convergence of $\widehat{\boldgreek\beta}_n(\alpha)$ process 
%in linear model, $\alpha_n^*<\alpha<1-\alpha_n^*$,  %is studied in \cite{Jur_Slopes}. 
%The asymptotic linearity can be apparently proven under various conditions, which should be verified for the specific real life problem to be solved. 
%in the literature;  the stronger results 
% impose the stronger conditions on distribution function $F$ and on the regression matrix.\\

Hence following  \cite{KoulSaleh1995}, we assume the conditions (\ref{moments_of_u}), (\ref{35}), ({35a})
and assume that $F$ satisfies the conditions {\bf(F1)} and {\bf(F2)}, while $\mathbf X_n$ satisfies the conditions {\bf(X1)}--{\bf(X3)}. 
Let  $\hat{\mathbf a}_{n\alpha}(\mathbf y_n-\mathbf  X_n\mathbf b)$ denote the vector of AR rank scores calculated for the residuals 
$\tilde{\mathbf y}_{n,\mathbf b}=\mathbf y_n-\mathbf  X_n\mathbf b$ for  $\mathbf b\in \R^s$, and let $\mathbf S_n$ be the linear 
regression rank scores statistic given in (\ref{37}). Then $\mathbf S_n$ exhibits the following asymptotic linearity in regression parameter:
\begin{eqnarray}\label{2.12}
&&\sup_{\|{\mathbf t}\|\leq C}\Big\{\|\mathbf S_n(\mathbf y_n-\mathbf X_n\boldgreek\beta-n^{-1/2}\mathbf  X_n\mathbf t)
-\mathbf S_n(\mathbf y_n-\mathbf X_n\boldgreek\beta)
+f(F^{-1}(\alpha))\mathbf Q_n\mathbf t\|\Big\}=o_p(1)\nonumber\\
&& \mbox{ as } \; \ny, \; \mbox{ for any fixed } \;  C>0.
%\left[\mathbf X_n-\bar{\mathbf X}_n\right]^{\top}\mathbf X_n\mathbf b\|\right\}=o_p(1)
\end{eqnarray}
This in turn implies that the Jaeckel criterion is  approximated by a quadratic function,  as $\ny$:
\begin{equation}\label{3.4}
\mathcal A_n=\sup_{\|{\mathbf t}\|\leq C}\Big|\mathcal D_n(\boldgreek\beta+n^{-1/2}\mathbf t)
+\mathbf t^{\top}\mathbf S_n(\mathbf y_n-\mathbf X_n\boldgreek\beta)
-\nfrac 12 f(F^{-1}(\alpha))\mathbf t^{\top}\mathbf Q_n\mathbf t\Big|\stackrel{p}{\rightarrow}0 %\; \mbox{ as }  \ny.\nonumber
\end{equation}
 Hence  the estimator $\widetilde{\boldgreek\beta}_{nRR}(\alpha)$ is  approximated by the minimizer of the quadratic approximation 
of $\mathcal D_n$,  as $\ny$;  this in turn gives its consistency and the asymptotic distribution.

\subsection{Consistency and asymptotic distribution\\ of the RR-estimator of $\boldgreek\beta$.}
Let  $\widetilde{\boldgreek\beta}_{nRR}(\alpha)$ be the estimator of $\boldgreek\beta$, defined in (\ref{Jaeckel}) %and (\ref{Jaeckel2}) 
for a fixed  $\alpha\in(0,1)$. 
Then $\widetilde{\boldgreek\beta}_{nRR}(\alpha)$ is invariant to the intercept and to the nuisance autoregression parameters $\phi_1,\ldots,\phi_p$.
Moreover, the criterion $\mathcal D_n(\mathbf b)$ is continuous, convex and piecewise linear function of $\mathbf b\in \R^s$ (see \cite{Jaeckel1972},
 \cite{KoulSaleh1995}, \cite{Jur1992a}). The asymptotic behavior of $\widetilde{\boldgreek\beta}_{nRR}(\alpha)$ is derived with the aid of
the asymptotic linearity of the linear autoregression rank statistics $\mathbf S_n$ in the regression parameter, 
leading to the asymptotic quadraticity of $\mathcal D_n$ [(\ref{2.12}) and (\ref{3.4})]. This leads to the  asymptotic behavior of the RR-estimator, 
described in the following theorem: 
\begin{THE} Under the conditions (\ref{moments_of_u}), (\ref{35}), ({35a}), {\bf(F1)} and {\bf(F2)},  {\bf(X1)}--{\bf(X3)},
 $\widetilde{\boldgreek\beta}_{nRR}(\alpha)$ is a consistent estimator of $\boldgreek\beta$,
invariant with respect to the autoregression of model errors. As $\ny,$ the asymptotic distribution of 
$\left\{n^{1/2}\left(\widetilde{\boldgreek\beta}_{nRR}(\alpha)-\boldgreek\beta\right)\right\}$ is normal
\begin{equation}\label{3.3}  
\mathcal N_s\left(\mathbf 0, \sigma^2\mathbf Q^{-1}\right) 
\mbox{ with  } \;  \sigma^2=\frac{\alpha(1-\alpha)}{f^2(F^{-1}(\alpha))}.
\end{equation}
\end{THE}
\textbf{Proof. } Denote 
\begin{equation}\label{3.8}
\mathbf T_n(\alpha)=n^{1/2}\left(\widetilde{\boldgreek\beta}_{nRR}(\alpha)-\boldgreek\beta\right).
\end{equation}
If  $\widetilde{\boldgreek\beta}_{nRR}(\alpha)$ minimizes $\mathcal D_n(\mathbf b),$ then  $\mathbf T_n(\alpha)$
minimizes  the convex function
\begin{equation}\label{3.9}
\mathcal D_n^{\ast}(n^{-1/2}\mathbf t)=\mathcal D_n(\boldgreek\beta+n^{-1/2}\mathbf t)-\mathcal D_n(\boldgreek\beta) 
\end{equation}
with respect to $\mathbf t\in \R^s.$ 
It follows from (\ref{3.4}) that
%\begin{equation}\label
\begin{equation}\label{3.10}
%\min_{\|\mathbf t\|\leq C}
\mathcal D_n^{\ast}(\boldgreek\beta+n^{-1/2}\mathbf t) =
%\min_{\|\mathbf t\|\leq C}
\left\{\nfrac 12 f(F^{-1}(\alpha)) \mathbf t^{\top}\mathbf Q_n\mathbf t
-\mathbf t^{\top}\mathbf S_n(\mathbf y_n-\mathbf X_n\boldgreek\beta) \right\}+o_p(1)
\end{equation}
uniformly for $\|\mathbf t\|\leq C$. The right-hand side of (\ref{3.10}) is minimized by %and the minimizer of the right-hand side of (\ref{3.10}) is
\begin{eqnarray}\label{3.11}
\mathbf V_n(\alpha)&=&\arg\min_{\mathbf t\in\R^s}\left\{\nfrac 12 f(F^{-1}(\alpha)) \mathbf t^{\top}\mathbf Q_n\mathbf t
-\mathbf t^{\top}\mathbf S_n(\mathbf y_n-\mathbf X_n\boldgreek\beta) \right\}\nonumber\\
                   &=& [f(F^{-1}(\alpha))]^{-1}\mathbf Q_n^{-1}\mathbf S_n(\mathbf y_n-\mathbf X_n\boldgreek\beta)=O_p(1)
\end{eqnarray}
and %Then
\begin{eqnarray}\label{3.10a}
&&\min_{\mathbf t\in\R^s}\left\{\nfrac 12 f(F^{-1}(\alpha)) \mathbf t^{\top}\mathbf Q_n\mathbf t
-\mathbf t^{\top}\mathbf S_n(\mathbf y_n-\mathbf X_n\boldgreek\beta) \right\}\nonumber\\
&&=-1/(2f(F^{-1}(\alpha))\mathbf S_n^{\top}
(\mathbf y_n-\mathbf X_n\boldgreek\beta) \mathbf Q_n^{-1}\mathbf S_n(\mathbf y_n-\mathbf X_n\boldgreek\beta).%\nonumber 
\end{eqnarray}

%Then
%\begin{equation}\label{3.14}
%\sup_{\|\mathbf t\|\leq C}\Big |\mathcal D^*_n(\boldgreek\beta+n^{-1/2}\mathbf t)
%+\nfrac 12f(F^{-1}(\alpha))%\{(\mathbf t-\mathbf V_n)^{\top}\mathbf Q_n\{(\mathbf t-\mathbf V_n)-
%\mathbf t^{\top}\mathbf Q_n\mathbf t_n\}\Big| \stackrel{p}{\rightarrow}0 \; \mbox{ as } \ny.\nonumber
%\end{equation} 
%and specifically,  as $\ny$
%\begin{equation}\label{3.16}
%\mathcal D_n(\boldgreek\beta+\mathbf V_n(\alpha))-\mathcal D_n(\boldgreek\beta)=
%-\nfrac 12f(F^{-1}(\alpha))\mathbf V_n^{\top}(\alpha)\mathbf Q_n\mathbf V_n(\alpha)+o_p(1)  
%\end{equation}
Following the steps in \cite{Jur1992a},
or in \cite{Pollard} and \cite{GJKP}, we want to show that $\|\mathbf T_n(\alpha)-\mathbf V_n(\alpha)\|=o_p(1)$. Consider the ball ${\mathcal B}_n$
with center $\mathbf V_n(\alpha)$ and radius $\delta>0$, and  follow the behavior of $\|\mathbf t-\mathbf V_n\|$ for $\mathbf t$ inside 
and outside $\mathcal B_n$ separately. For $\mathbf t\in \mathcal B_n$, 
\begin{equation}\label{3.18} 
\|\mathbf t\|\leq \|\mathbf t-\mathbf V_n\|+\|\mathbf V_n\|\leq \delta+K
\end{equation}
for some $K>0$, with probability exceeding $1-\varepsilon$ for $n\geq n_0,$ hence  (\ref{3.10}) applies.
%\begin{equation}\label{3.23}
%P(\|\mathbf T_n-\mathbf V_n\|\leq \delta)\rightarrow 1\quad \mbox{ for any fixed } \;  \delta>0 \quad \mbox{ as } \; \ny.
%\end{equation}
If $\mathbf t$ lies outside $\mathcal B_n,$ then $\mathbf t=\mathbf V_n+ k\boldgreek\xi$ with $k>\delta$ and $\|\boldgreek\xi\|=1.$ Let $\mathbf t^*$
be the boundary point of $\mathcal B_n$ which lies on on the line connecting  $\mathbf T_n$ and $\mathbf V_n$, thus 
$\mathbf t^*=\mathbf V_n+\delta\boldgreek\xi$. Then, because $\mathcal D^*_n$ is convex and 
$\mathbf t^*=\nfrac{\delta}{k}\mathbf t+(1-\nfrac{\delta}{k})\mathbf V_n$ satisfies (\ref{3.4}), we can conclude
\begin{eqnarray}\label{3.20}
\nfrac{\delta}{k}\mathbf \mathcal D^*_n(\mathbf t)+(1-\nfrac{\delta}{k})\mathcal D^*_n(\mathbf V_n)\geq \mathcal D^*_n(\mathbf t^*)
\geq \nfrac{1}{2}\lambda_0\delta^2+\mathcal D_n^*(\mathbf V_n)-2\mathcal A_n.
\end{eqnarray}
where $\lambda_0$ is the minimal eigenvalue of $\mathbf Q.$ Hence,
\begin{equation}\label{3.21}
\inf_{\|\mathbf t-\mathbf V_n\| \geq \delta}\mathcal D^*(\bf t)\geq \mathcal D_n^*(\bf V_n)
+\nfrac{k}{\delta}\left(\nfrac 12 \lambda_0\delta^2-2\mathcal A_n\right). 
\end{equation}
Hence, given $\delta>0$ and $\varepsilon>0,$ there exist $n_0$ and $\tau>0$, such that for $n\geq n_0$
\begin{equation}\label{3.22}
P\left\{\inf_{\|\bf t-\bf V_n\|\geq \delta}  \mathcal D_n^*(\bf t)>\mathcal D_n^*(\mathbf V_n+\tau)\right\}>1-\varepsilon,
\end{equation}
thus
\begin{equation}\label{3.23a}
P\left(\|\mathbf T_n-\mathbf V_n\|)\leq \delta\right)\rightarrow 1 \; \mbox{ as } \ny \; \mbox{ for any fixed } \; \delta >0. 
\end{equation}
For the asymptotic distribution of $\mathbf V_n$, hence also of $\mathbf T_n$, we refer e.g. to \cite{Hallin1999}. 
\hfill $\Box$

\section{Simulation study under autoregressive errors}
\setcounter{equation}{0}

To further evaluate the performance of the proposed RR-estimator, we conducted a simulation study in a linear regression model with autoregressive errors. We compared its performance to the classical Ordinary Least Squares (OLS) estimator and the median Quantile Regression (QR) estimator. The model considered was:

\begin{equation}
    y_t = \beta_0 + \beta_1 x_{t1} + \beta_2 x_{t2} + \beta_3 x_{t3} + \varepsilon_t,
\end{equation}
where the errors follow an AR(1) process:
\begin{equation}
    \varepsilon_t = \phi \varepsilon_{t-1} + u_t,
\end{equation}
with $\phi = 0.7$ and innovations $u_t$ generated from three different distributions:
\begin{enumerate}
    \item Standard normal: $u_t \sim \mathcal{N}(0, 1)$,
    \item Contaminated normal: $u_t \sim 0.9\, \mathcal{N}(0, 1) + 0.1\, \mathcal{N}(0, 25)$,
    \item Cauchy: $u_t \sim \text{Cauchy}(0, 1)$.
\end{enumerate}

The true parameter values were set as $\beta_0 = 1$, $\beta_1 = 2$, $\beta_2 = -1$, and $\beta_3 = 0.5$. The covariates $x_{tj}$ were generated independently from $\mathcal{N}(0,1)$ for $j = 1,2,3$, and the sample size was fixed at $n = 100$. Each setting was replicated $1000$ times.

For each scenario, we computed the average bias and Root Mean Squared Error (RMSE) of the estimated coefficients across the 1000 replications. Table~\ref{tab:sim_results_detailed} summarizes the results.

\begin{table}[h!]
\centering
\caption{RMSE and Bias for parameter estimates under different error distributions and estimation methods.}
\label{tab:sim_results_detailed}
\begin{tabular}{|l|c|c|c|}
\hline
\textbf{Distribution - Metric (Var)} & \textbf{OLS} & \textbf{QR} & \textbf{RR} \\
\hline
Normal - RMSE (Intercept) & 0.3493 & 0.4296 & 0.4306 \\
Normal - RMSE (x1) & 0.416 & 0.4929 & 0.4934 \\
Normal - RMSE (x2) & 0.3879 & 0.4925 & 0.4896 \\
Normal - Bias (Intercept) & 0.0168 & 0.0092 & 0.0092 \\
Normal - Bias (x1) & -0.026 & 0.0038 & 0.0002 \\
Normal - Bias (x2) & -0.0133 & -0.0279 & -0.025 \\
\hline
Contaminated - RMSE (Intercept) & 1.1769 & 0.6351 & 0.6339 \\
Contaminated - RMSE (x1) & 1.2922 & 0.707 & 0.7052 \\
Contaminated - RMSE (x2) & 1.3328 & 0.6931 & 0.6898 \\
Contaminated - Bias (Intercept) & -0.0094 & 0.0643 & 0.0618 \\
Contaminated - Bias (x1) & -0.0658 & -0.0875 & -0.0866 \\
Contaminated - Bias (x2) & -0.0039 & -0.0328 & -0.0306 \\
\hline
Cauchy - RMSE (Intercept) & 78.6605 & 0.9871 & 0.9841 \\
Cauchy - RMSE (x1) & 67.0591 & 1.1694 & 1.1595 \\
Cauchy - RMSE (x2) & 60.8429 & 1.1067 & 1.1063 \\
Cauchy - Bias (Intercept) & 4.791 & 0.0356 & 0.0347 \\
Cauchy - Bias (x1) & -3.1077 & 0.0403 & 0.0409 \\
Cauchy - Bias (x2) & -0.6391 & 0.0372 & 0.0373 \\
\hline
\end{tabular}
\end{table}

Figures~\ref{fig:rmse_plot} and~\ref{fig:bias_plot} provide visual summaries of the RMSE and bias, respectively, across all three estimators and distributions.

\begin{figure}[h!]
\centering
\includegraphics[width=0.85\textwidth]{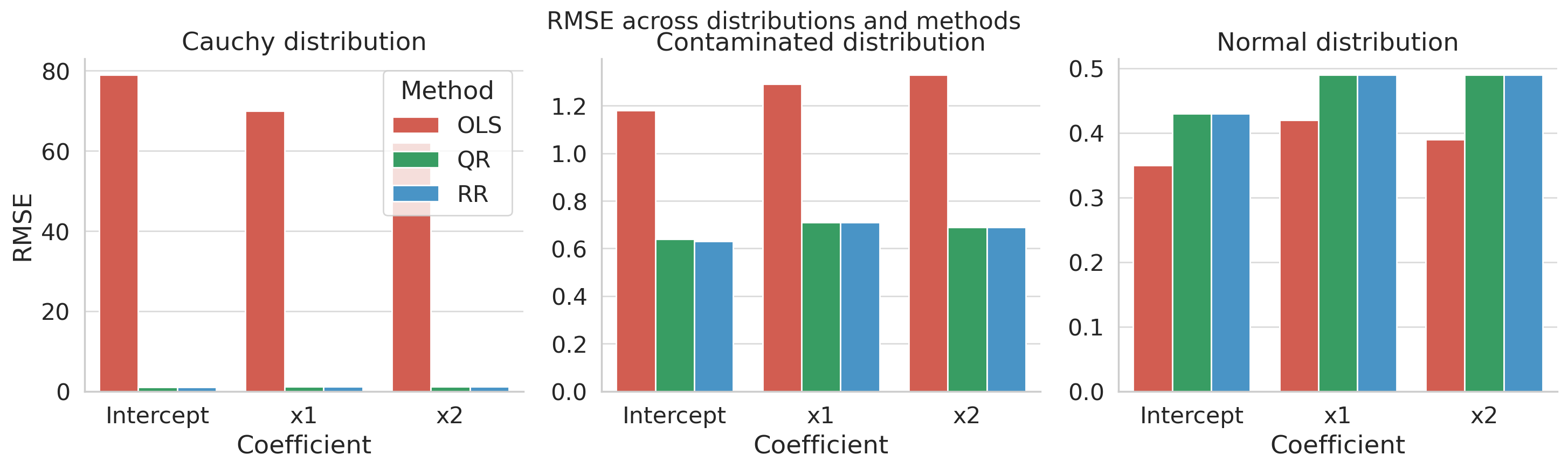}
\caption{Root Mean Squared Error (RMSE) for OLS, QR ($\tau=0.5$), and RR ($\alpha=0.5$) across distributions.}
\label{fig:rmse_plot}
\end{figure}

\begin{figure}[h!]
\centering
\includegraphics[width=0.85\textwidth]{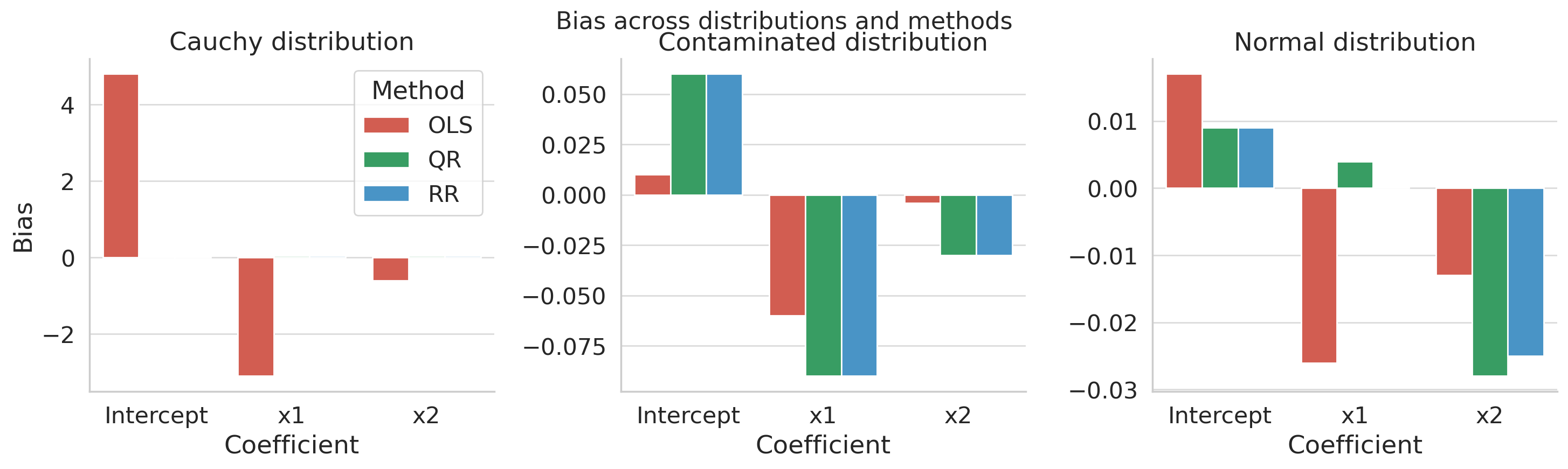}
\caption{Bias for OLS, QR ($\tau=0.5$), and RR ($\alpha=0.5$) across distributions.}
\label{fig:bias_plot}
\end{figure}

These results demonstrate that while OLS performs well under Gaussian errors, it is highly sensitive to contamination and heavy tails. Both QR and RR estimators offer substantially improved robustness in such cases. The RR-estimator often yields slightly lower RMSE than QR, confirming its competitiveness in scenarios with autoregressive and nonstandard error structures.

\section{Numerical illustration using real estate data}
\setcounter{equation}{0}

To complement the theoretical properties of the RR-estimator, we now demonstrate its behavior on a real-world data set. Specifically, we analyze the well-known real estate valuation data, originally collected by the Statistical Office of Taiwan and publicly available as ``Real estate valuation data set.xlsx''~\cite{uci_realestate}. The data consist of 414 observations and include explanatory variables such as the age of the house, the distance to the nearest MRT station, and the number of convenience stores within walking distance. The response variable is the price of a housing unit per square meter (expressed in NT\$10000/m$^2$).

We consider the following linear regression model:
\begin{equation}
\texttt{price}_t = \beta_0 + \beta_1\cdot\texttt{age}_t + \beta_2\cdot\texttt{distance}_t + \beta_3\cdot\texttt{shops}_t + \varepsilon_t,
\end{equation}
where the error term $\varepsilon_t$ may exhibit an autoregressive structure or other forms of deviation from standard assumptions. The parameters were estimated using three different approaches:
\begin{enumerate}
\item Ordinary Least Squares (OLS),
\item Quantile Regression (QR) for $\tau = 0.5$,
\item Rank-based RR-estimator with $\alpha = 0.5$.
\end{enumerate}

The estimation was performed using the R packages \texttt{quantreg} \cite{koenker_quantreg} and \texttt{readxl}, with the RR-estimator computed by exploiting the dual solution of the quantile regression problem with $\tau = -1$.

\begin{table}[h!]
\centering
\caption{Comparison of regression parameter estimates obtained by OLS, median quantile regression (QR), and rank-based RR estimation.}
\label{tab:reg_compare}
\begin{tabular}{lrrr}
\hline
 & OLS & QR ($\tau=0.5$) & RR ($\alpha=0.5$) \\
\hline
Intercept & 42.9773 & 40.5410 & 40.7748 \\
Age       & -0.2529 & -0.2464 & -0.2519 \\
Distance  & -0.0054 & -0.0053 & -0.0053 \\
Shops     &  1.2974 &  1.4381 &  1.4202 \\
\hline
\end{tabular}
\end{table}

Table~\ref{tab:reg_compare} presents the parameter estimates obtained by the three methods. The estimates are relatively stable across all methods, indicating robustness of the linear relationship. However, some differences can be observed.

The intercept shows a noticeable downward shift when switching from OLS to QR and RR, which is expected given that QR and RR are less sensitive to potential skewness or outliers in the response variable. The slopes for all covariates remain close across the methods, with the RR estimates being generally closer to QR than to OLS, reflecting their shared robustness properties. In particular, the estimate for the number of convenience stores is slightly higher for QR and RR compared to OLS, suggesting that the influence of this variable may be more pronounced in the central part of the conditional distribution.

\begin{figure}[h!]
\centering
\includegraphics[width=0.9\textwidth]{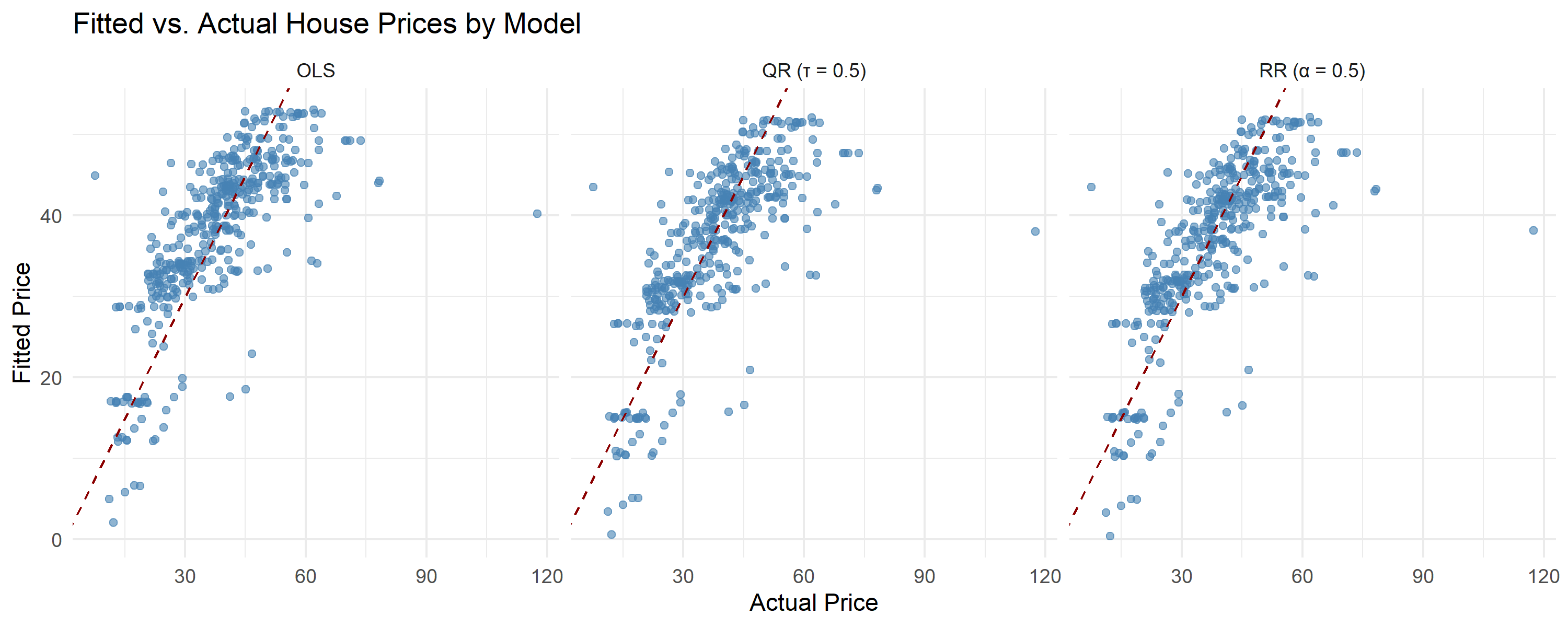}
\caption{Fitted vs. actual house prices shown separately for OLS, QR ($\tau=0.5$), and RR ($\alpha=0.5$). Each subplot compares predicted and actual prices with a reference identity line.}
\label{fig:fitted_vs_actual_panel}
\end{figure}

\begin{figure}[h!]
\centering
\includegraphics[width=0.75\textwidth]{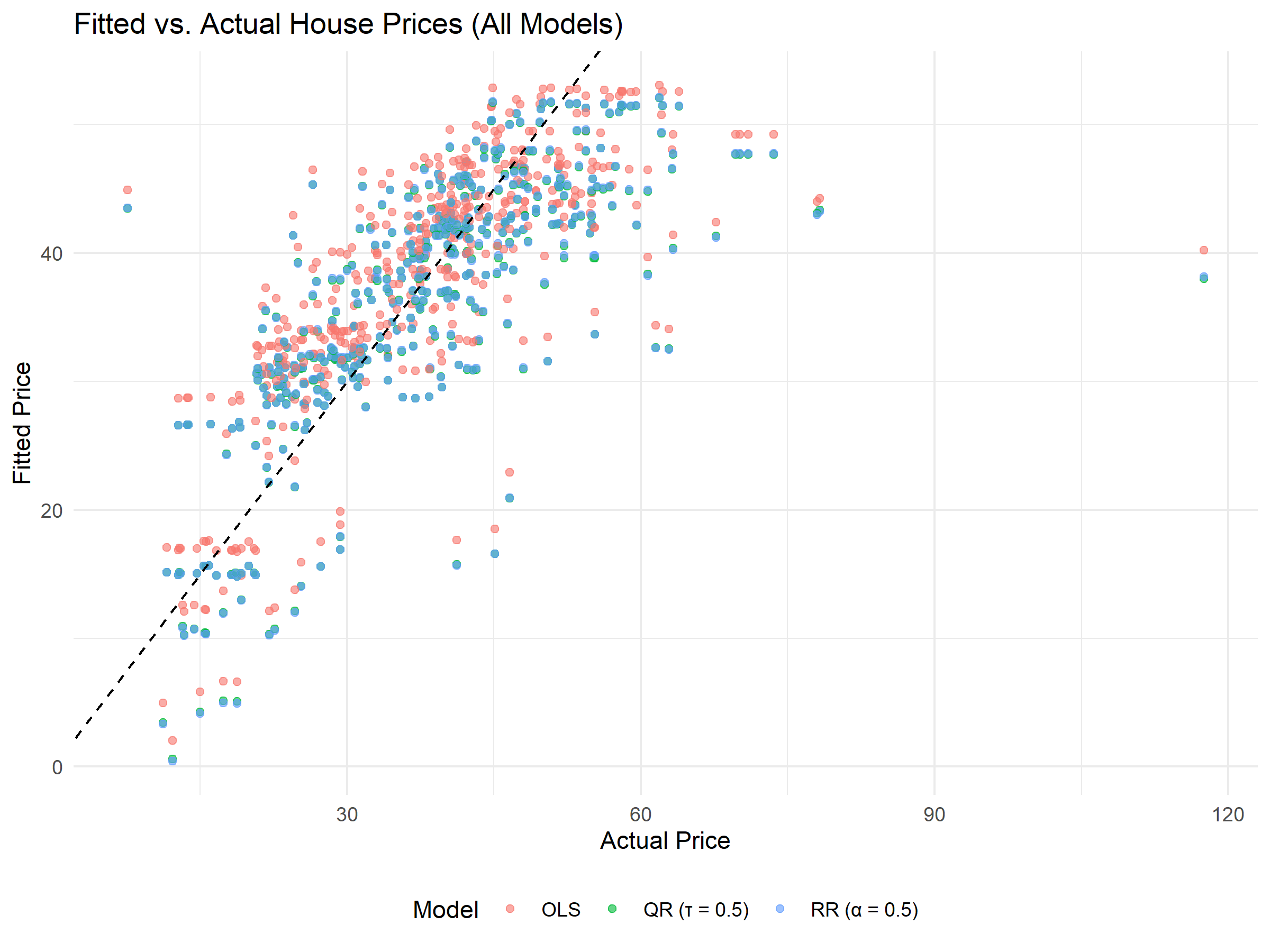}
\caption{Fitted vs. actual house prices shown in a single scatterplot with different colors for OLS, QR ($\tau=0.5$), and RR ($\alpha=0.5$). The dashed line denotes the identity line.}
\label{fig:fitted_vs_actual_overlay}
\end{figure}

Figures~\ref{fig:fitted_vs_actual_panel} and~\ref{fig:fitted_vs_actual_overlay} provide a graphical comparison between the observed and fitted house prices. The faceted version helps assess each method's fit individually, while the overlaid version allows direct comparison. All three models show good alignment with the identity line, indicating reasonable predictive accuracy. However, the QR and RR methods tend to perform slightly better in the presence of extreme values, as expected due to their robustness.

The RR and QR estimates demonstrated improved performance over OLS in terms of both RMSE and MAE (Mean Absolute Error) on the test set, highlighting their advantage in cases with non-normal or heteroskedastic errors.

These findings support the usefulness of the RR-estimator in situations where the error structure might deviate from the classical assumptions, while maintaining reasonable efficiency and interpretability.

\section{Conclusion}
\setcounter{equation}{0}

In this paper, we have introduced a novel rank-based estimator (RR-estimator) for the regression parameters in a linear model with autoregressive errors. Unlike classical methods such as Ordinary Least Squares (OLS), the RR-estimator remains robust under heavy-tailed error distributions, contamination, and nuisance autoregressive dependencies. In contrast to median Quantile Regression (QR), the RR-estimator is specifically adapted to the presence of serial correlation, providing invariance not only to distributional deviations but also to structural error dependence.

The RR-estimator inherits desirable theoretical properties, including asymptotic normality and efficiency under a broad class of error distributions. Practically, it offers greater stability in contaminated or heavy-tailed settings without requiring explicit modeling of the autoregression. Compared to QR, it provides an elegant alternative when serial dependence is suspected but the autoregressive structure is not explicitly specified or reliably estimated.

Our simulation studies and real-data applications demonstrate that the RR-estimator achieves comparable efficiency to QR under normal errors and superior robustness under contamination or heavy-tailed innovations. It consistently outperforms OLS, particularly in non-ideal conditions that frequently arise in applied fields such as hydrology and economics.

\end{document}